\documentclass[letterpaper,10pt]{IEEEtran}
\IEEEoverridecommandlockouts
\usepackage{mathtools}
\usepackage{ifpdf}
\usepackage{amsthm}
\usepackage{nicefrac}
\usepackage{cite}
\usepackage{url}
\theoremstyle{definition}

\newtheorem{lem}{Lemma}

\begin{document}
\title{Everlasting Secrecy by Exploiting Non-Idealities of the Eavesdropper's Receiver}
\date{}
\author{
\IEEEauthorblockN{Azadeh Sheikholeslami,\IEEEmembership{ Student Member, IEEE}, Dennis Goeckel,\IEEEmembership{ Fellow, IEEE} and Hossein Pishro-Nik, \IEEEmembership{ Member, IEEE}
}\\
\thanks{Manuscript received September 15, 2012; revised January 19, 2013; accepted April 15, 2013. This work appears, in part, at the 2012 Allerton Conference on Communication, Control, and Computing. This work was supported by the National Science Foundation under grants CNS-0905349, ECCS-0725616, and CIF-1249275.}
\thanks{Authors are with the Electrical and Computer Engineering  Department, University of Massachusetts,
Amherst, MA (emails:\{sheikholesla,goeckel,pishro\}@ecs.umass.edu).}}
\maketitle
\begin{abstract}
Secure communication over a memoryless wiretap channel in the presence of a passive eavesdropper is considered.  Traditional information-theoretic security methods require an advantage for the main channel over the eavesdropper channel to achieve a positive
secrecy rate, which in general cannot be guaranteed in wireless systems.  Here, we exploit
the non-linear conversion operation in the eavesdropper's receiver to obtain the desired advantage - even when the eavesdropper has perfect access to the transmitted signal at the input to their receiver.  The basic idea is to employ an ephemeral cryptographic key to force the eavesdropper to conduct two operations, at least one of which is non-linear, in a different order than the desired recipient. 
Since non-linear operations are not necessarily commutative, the desired advantage can be obtained and information-theoretic secrecy achieved even if the eavesdropper is given the cryptographic key immediately upon transmission completion.  
In essence, the lack of knowledge of the key during the short transmission time inhibits the recording of the signal in such a way that the secret information can never be extracted from it. The achievable secrecy rates for different countermeasures that the eavesdropper might employ are evaluated. It is shown that even in the case of an  eavesdropper with uniformly better conditions (channel and receiver quality) than the intended recipient, a positive secrecy rate can be achieved.
\end{abstract}
\section{Introduction} \label{sec:1}
Wireless networks, due to their broadcast nature, are vulnerable to being 
overheard, and hence security is a primary concern.  The 
standard method of providing security against eavesdroppers is to 
encrypt the information so that it is beyond the eavesdropper's 
computational capabilities to decrypt the message \cite{stinson2006cryptography}; 
however, the vulnerability shown by many implemented cryptographic schemes, 
the lack of a fundamental proof establishing the difficulty of the problem 
presented to the adversary, and the potential for transformative 
changes in computing motivate forms of security that are provably 
everlasting.  In particular, when a cryptographic scheme is employed, the 
adversary can  record the clean cypher and recover it later when the 
cryptographic algorithm is broken \cite{bensonverona}, which is not 
acceptable in sensitive applications requiring everlasting secrecy. 
The desire for such everlasting security motivates considering emerging 
information-theoretic approaches, where the eavesdropper is unable to extract 
from the received signal any information about the secret message. 

In 1949, Shannon introduced information-theoretic, or perfect, 
secrecy \cite{shannon1949communication}.  If the uncertainty of the 
message after seeing the cypher is equal to the uncertainty of the message 
before seeing the cypher, we have perfect secrecy without any condition on 
the eavesdropper's capabilities. Wyner later showed  that if the eavesdropper's channel is degraded
 with respect to the main channel, adding some randomness to the codebook 
allows the achievement of a positive secrecy rate \cite{wyner1975wire}.  Csisz{\'a}r and 
K{\"o}rner  extended the idea to more general cases, where the eavesdropper's 
channel is not necessarily degraded with respect to the main channel, but 
it must be ``more noisy'' or ``less capable'' than the main channel 
\cite{csiszar1978broadcast}.  When such an advantage does not exist,
one can turn to approaches based on ``public discussion'' \cite{maurer1993secret,ahlswede1993common},
but these approaches, while they could be used to generate an
information-theoretically secure one-time pad, are generally envisioned for 
secret key agreement to support a cryptographic approach 
\cite[Chapter 7.4]{bloch2011physical} rather than  one-way secret communication. 
We will show later the relation between our proposed scheme and public discussion, noting, in particular, that the proposed scheme can be used in conjunction with public discussion when appropriate.

Consequently, the desirable situation for achieving information-theoretic 
secrecy is to have a better channel from the transmitter to the intended receiver 
than that from the transmitter to the eavesdropper.
  However, this is not always 
guaranteed, particularly in wireless systems where the eavesdropper can have a 
large advantage over the intended receiver.
In the case of a passive adversary, 
the eavesdropper can be very close to the transmitter or it can use a directional 
antenna to improve its received signal, while there is often no way for the legitimate
nodes to know the eavesdropper's location or its channel state information.   
Recent authors have considered approaches that relax the need for assumptions on 
Eve's location or channel in one-way systems. 
For cases when the eavesdropper location is unknown (which means the case of a ``near
Eve'' must be considered), approaches largely based on the cooperative jamming approach of \cite{negi2005secret} and \cite{goel2005secret} have been considered \cite{he2008two,lai2007cooperative}.
However, all of these approaches require either
multiple antennas, helper nodes, and/or fading (for example, \cite{gopala2008secrecy,Dennis2011JSAC,sheikholeslami2012physical}), and many are susceptible to attacks
such as pointing directive antennas at one or both communicating parties.

For a one-way scenario with a single antenna where Bob's channel is worse than Eve's, 
Cachin and Maurer \cite{cachin1997unconditional} exploited the realizability of hardware 
to consider the case of everlasting security, as is our interest. In particular, they 
introduced the ``bounded storage model''  in which the receiver
cannot store the information it would need to eventually break the cypher. This novel 
approach suffers from two shortcomings: (1) by Moore's Law (see NAND scaling plot 
at \cite{kuchibhatla2010imft}), the density of memories increases at an exponential rate; 
(2) memories can be stacked arbitrarily subject only to (very) large space limitations.
Hence, although the bounded storage
 model is a viable approach to everlasting security, 
it is difficult to pick a memory size beyond which it will be effective, making its 
employment for secret wireless communication difficult.
Rather than attacking the memory in the receiver back-end, our contention is that one should instead consider attacking  the receiver front-end and analog-to-digital (A/D) conversion process, where technology progresses slowly and there exist well-known 
techniques for severely handicapping the component. And, unlike memory, A/D's cannot be
stacked arbitrarily, as clock jitter prevents the timing required for bit detection; in fact, 
high-quality A/D's already employ parallelization to the limit of the jitter. 
And, importantly from a long-term perspective, there is a fundamental bound 
on the ability to perform A/D conversion \cite{krone2009fundamental,krone2010fundamental}. 
\begin{figure}
\begin{center}
 \includegraphics[width=.5\textwidth]{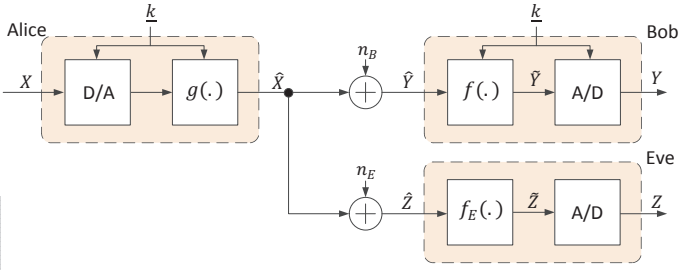}
 \end{center}
 \caption{The message $X$ is observed at Bob and Eve through the transmitter, the AWGN channels with different noise variances, and their respective receivers with (possibly nonlinear) functions  $g(.)$, $f(.)$, and $f_E(.)$. The sequence $\underline{k}$ is a cryptographic key shared by Alice and Bob, which
is assumed to be obtained by Eve immediately \textit{after} she has recorded $Z$.}
 \label{fig:p}
 \end{figure}
Consider the  channel model shown in Figure \ref{fig:p}, which reflects the understanding
that in an adversarial game in modern communication systems, it is the interference effects 
on wideband receiver front-ends rather than the baseband processing that is the significant 
detriment \cite{Hashemi2011RFIC}.  In particular, the signal is subject to a variety of 
distortions due to the RF front-end of the receiver and the analog-to-digital conversion. 
 A large interferer, even if it is orthogonal to the signal of interest and thus (supposedly)
easily rejected by  baseband processing, can saturate the receiver front-end, leading 
to nonlinearities, and, of particular interest here, reducing the receiver's 
dynamic range (i.e. resolution) significantly.

The primary focus of this paper is to exploit the receiver processing effects for security. 
In particular, based on a pre-shared key between Alice and Bob that only needs to be
kept secret for the duration of the wireless transmission (i.e. it can be given to 
Eve immediately afterward), we consider how inserting intentional (but known to
Bob) distortion on the transmitted signal can provide information-theoretic
security.  In particular, since Bob knows the distortion, he can undo its
effect before his A/D, whereas Eve must store the signal and try to compensate
for the distortion after her A/D.  Since the A/D is necessarily non-linear, 
the operations are not necessarily commutative and there is the potential for
information-theoretic security.  This paper introduces this idea and initiates
its investigation.

As a first example, we perform a rapid power modulation between two vastly different power levels at the transmitter and put the reciprocal of that
power gain before Bob's A/D.  
In particular, cellular (and other) networks
usually have significantly more power available for users at many locations
than their lowest data rate requires for successful transmission. 
For example, users near a base station in a cellular system have the capability
to transmit significantly more power than the minimum required to convey
a high-quality voice signal.  
Hence, a secure communication system to cover
a restricted area (e.g. a company's building) built on analogous link budgets
to cellular technology would have the capability to transmit excess power to
enable secure communication, as follows.  
Suppose Alice employs an ephemeral
cryptographic key known only to her and Bob to rapidly modulate her transmit
power between the minimum required for successful transmission and the maximum
available from her radio.  This power modulation can be done quite rapidly,
as modern power amplifiers can easily have their power switched at high bandwidths
\cite{kahn1952single}\cite[Chapter~7]{kenington2000high}.  
Bob, since he knows the key, places a gain before
his A/D that changes rapidly in concert with the transmitted power to ensure
that the received signal is matched to the range of the A/D.  
Since the power
can be changed every symbol, Eve cannot use any type
of automatic gain control (AGC) loop and is left trying to select a gain that trades off resolution and the probability of overflow of her A/D.
By exploiting the resulting distortion, information-theoretic secrecy can be obtained, even if Eve is given the key immediately after message transmission.

The rest of paper is as follows. Section \ref{sec:2} describes the system model, metrics, and the proposed idea in detail. In Section \ref{sec:3}, the proposed method is applied to settings with noisy channels and noiseless channels, respectively, to find achievable secrecy rates in each case, and an asymptotic analysis of the proposed method is provided. 
In Section \ref{sec:4}, the results of numerical examples for  various   realizations of the system are presented. 
Conclusions and ideas for future work  are discussed in Section \ref{sec:7}.

\section{System Model and Approach}\label{sec:2}
\subsection{System Model and Metric}
We consider a simple wiretap channel, which consists of a transmitter, Alice, a receiver, Bob, and an eavesdropper, Eve.
 Eve is a passive eavesdropper, i.e. she just tries to obtain as much information as possible to recover the message that Alice sends and she does not attempt to actively  thwart (i.e. via jamming, signal insertion) the legitimate nodes. 
 Therefore, the location and channel state information of Eve can be difficult to obtain and thus is assumed  unknown to the legitimate nodes.

 We assume that Alice and Bob either pre-share a (very) short initial  key or that they employ a standard  key agreement scheme (e.g. Diffie-Hellman \cite{diffie1976new}, which is very efficient in passive environments) to generate a shared key. 
This initial key will be used to generate a very long key-sequence by using a standard cryptographic method such as AES in counter mode (CTR). Considering the fact that for each $2^{38}$ bits of the key-sequence, a 96-bit new initial vector (IV) or a 128-bit new initial key must be sent from Alice to Bob \cite[Chapter 5]{paar2010understanding}, the secrecy rate overhead that this key (or IV) exchange imposes is at most $128/2^{38}=2^{-29}$, which is negligible.
Another method is to use standard methods that are specifically designed for generating stream-ciphers, such as Trivium (more methods can be found in \cite{robshaw2008new}), which can generate $2^{64}$ bits of key-sequence for a 80-bit key and a 80-bit IV. Thus, the rate overhead that Trivium places on our scheme will be $80/2^{64}<2^{-55}$, which is negligible.

  By using these cryptographic algorithms to perform key-expansion, we assume that Eve
cannot recover the initial key before the key renewal and during the transmission
period, i.e. we assume that the computational power of Eve during the time of transmission is not unlimited. However, our system design only employs the key ephemerally. In fact,  we assume (pessimistically) that Eve is handed the full key (and not just the initial key) as soon as transmission is complete.  
Hence, unlike cryptography, even if the encryption system is broken later, the eavesdropper obtains access to an unlimited computational power, or other forms of computation such as quantum computers are implemented, Eve will not have enough information to recover the secret message. 

 We consider a memoryless one-way communication system, and assume that both Bob and Eve are at a unit distance from the transmitter by including variations in the path-loss in the noise variance. 
 Thus, the channel gain of both channels is unity and both channels experience additive white Gaussian noise (AWGN). 
 Let $n_B$ and $n_E$ denote the zero-mean noise processes  at  Bob's  and Eve's receivers with variances $\sigma_B^2$ and ${\sigma_E}^2$, respectively.  
 Let $\hat{X}$ denote the input of both channels, $\hat{Y}$ denote  the received signal at Bob's receiver, and $\hat{Z}$  denote  the received signal at Eve's receiver. 
 The signal at Bob's receiver is:
 \[\hat{Y}=\hat{X}+n_B,\]
 and the signal at Eve's receiver is:
 \[\hat{Z}=\hat{X}+n_E.\]
  We assume that location of Alice is known to Eve. Also, Alice knows either Eve's location, or in the case that she does not know Eve's location, she  sets a value that works over a set of locations (for example, the minimum  possible distance between Alice and Eve). If the location of Eve is completely unknown, Eve's distance can assumed to be zero and, as will be shown in Section \ref{sec:4}, the legitimate nodes will still be able to obtain a positive secrecy rate by using the proposed scheme.
	
 Both Bob and Eve employ high precision uniform analog-to-digital converters. 
 The effect of the A/D on the received signal (quantization error) is modeled by a quantization noise due to the limitation in the size of each quantization level, and a clipping function due to the quantizer's overflow. 
 The quantization noise in this case is (approximately) uniformly distributed \cite{widrow2008quantization}, so we will assume it is uniformly distributed throughout the paper. 
 For an \textit{m}-bit quantizer ($b=2^m$ gray levels) over the full dynamic range $[-l,l]$, two adjacent quantization levels are spaced by  $\delta={2l}/{b}$, and thus the quantization noise is uniformly distributed over an interval of length $\delta$. 
 Quantizer overflow happens when the amplitude of the received signal is greater than the quantizer's dynamic range, which can be modeled by a clipping function. We assume that Alice knows an upper bound on Eve's current A/D conversion ability (without any assumption on Eve's future A/D conversion capabilities).

Let $X$ denote the current code symbol, which we assume  is taken from a standard Gaussian codebook where each entry has variance $P$, i.e. $X\sim \mathcal{N}(0,P)$. 
Note that although the Gaussian codebook is optimal to achieve the secrecy capacity in the case of AWGN wiretap channels,  because we consider  quantization errors in our model, the Gaussian codebook is no longer optimum, implying that our results represent achievable rates but not upper bounds.
 
From \cite{bloch2008secrecy}, for an arbitrary stationary memoryless wiretap channel with arbitrary input and output alphabets, any secrecy rate 
 \[
\hat{R}_s< \max_{X\rightarrow YZ} [I(X;Y)-I(X;Z)] 
 \]
 is achievable. 
 
 Now, we define the following max-min criteria:
 \begin{equation}
R_s=\max_{s\in\mathcal{S}}\min_{s'\in \mathcal{S}'} \hat{R}_s(s,s')
\label{eq:101}
\end{equation}
where $\mathcal{S}'$ is the set of strategies that Eve can take during transmission, and $\mathcal{S}$ is the set of strategies that Alice can take.
Eve's problem is to find a strategy, $s'\in\mathcal{S}'$,  to modify her channel to minimize the secrecy rate. On the other hand, Alice's problem is to find a strategy, $s\in\mathcal{S}$, to modify the transmit signal to maximize this worst-case secrecy rate.

 When  cryptographic key expansion schemes are employed, the key-sequence is not quite memoryless. 
 But, based on the assumption that Eve cannot restrict the rest of the key sequence based on the observed symbols,  we assume independence. Hence, although in general the strategy taken by Eve is not memoryless, here  considering strategies with memory does not help her to increase the information-leakage;
  thus, we  restrict $\mathcal{S}'$ to memoryless strategies. 
  Further, we give  the key to Eve after completion of the transmission and show she cannot recover the lost information she would need to obtain the secret message from the recorded symbols.

\subsection{General Nonlinearity: Rough Analysis}
Our goal is to  consider how Alice and Bob can employ  bits of the
shared  key to modify their radios as shown in Figure \ref{fig:p} to gain (or maximize) an information-theoretic advantage. 
For now, assume that they insert general memoryless nonlinearities $g(.)$ at the transmitter and 
$f(.)=g^{-1}(.)$ at the receiver based on the key. 
Suppose  that Eve is able to obtain the key just after the transmission is finished; considering for the moment that she applies $g^{-1}(.)$ to $Z$, one sees how the security is (potentially) obtained: 
Bob sees $g(X)$ through $g^{-1}(.)$ and the A/D, whereas Eve sees those operations in reverse. 
Since nonlinear operations are not (necessarily) commutative, the signals are not the same and there is the potential for some form of information-theoretic security. 

Now, stepping back to allow Eve to use the long key sequence, $\underline{k}$, in whatever manner she wants after
she has recorded the transmission yields an illustrative information-theoretic  model. 
In particular, using the same random coding arguments as for fading channels,  consider a collection  of functions $\mathcal{G}$, from which $\underline{k}$ selects a function  $g(.)$ for each transmitted symbol; then,  the secrecy rate is:
\[
R_s=E_{g(.)}[I(X;Y|g(.))-I(X;Z|g(.))]
\]
Let us be pessimistic and assume $\sigma_E^2=0$. Furthermore, to get some insight, assume temporarily that 
$\sigma_B^2=0$, corresponding to a short-range situation which is not power-limited. For $\sigma_B^2=0$, $Y$ does not depend on $\underline{k}$ and thus using the approach for analyzing quantizers of \cite[pg. 251]{cover2006elements}, which is accurate at high resolution:
\begin{align*}
&R_s=E_{g(.)}[I(X;Y)-I(X;Z|g(.))]\\
&=E_{g(.)}[H(Y)-H(Y|X)-(H(Z|g(.))-H(Z|X,g(.)))]\\
&\approx E_{g(.)}[h(\tilde{Y})-\log(\delta)-(h(\tilde{Z}|g(.))-\log(\delta))]\\
&= E_{g(.)}[h(\tilde{Y})-h(\tilde{Z}|g(.))]\\
&=E_{g(.)}[h(X)-h({g(X)})]
\end{align*}
where $\tilde{Y}$ and $\tilde{Z}$ are the inputs to Bob and Eve's A/D converters, respectively.
It then becomes apparent that
the gain observed here for high-resolution A/D's at both Bob and Eve is a shaping gain between $X$ and $g(X)$.   Whereas we think of shaping gains
as tending to be relatively small (1.53 dB on the Gaussian
channel \cite{cald1990}), that is because the generally
considered gains are between the optimal (Gaussian) shaping
and a standard but reasonable (uniform) shaping.  In our design
scenario, if we are able to severely distort the signal, the gains
can become enormous.
We quickly caveat this conclusion by noting that the assumption $\sigma_B^2=0$ is critical, since those $g(.)$ which are most distorting can also cause significant ``noise enhancement'' on the channel from Alice to Bob. Hence, unless the noise is truly negligible (i.e. very short range communication), judgment should be reserved on the applicability of the technique until $\sigma_B^2\neq 0$ is considered in Section \ref{sec:3}.

\subsection{Rapid power modulation for secrecy}
 \begin{figure}
\begin{center}
 \includegraphics[width=.5\textwidth]{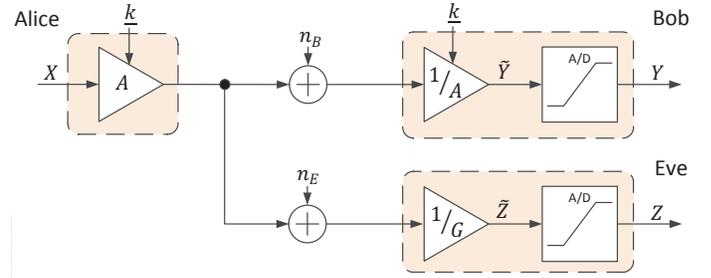}
 \end{center}
 \caption{Alice and Bob share a cryptographic key that determines the value of $A$ at each time instance. Eve puts a (possibly variable) gain before her A/D to decrease the A/D erasures and/or overflows and hence increase the information leakage.}
 \label{fig:p10}
 \end{figure}
For the rest of the paper, we simplify the operator $g(.)$ to a random gain to consider a practical architecture easily
implemented and discuss specific operating scenarios. 
Our goal is to  achieve a positive secrecy rate by confusing Eve's A/D. 
Throughout this paper we assume that Eve is able to employ just one A/D, and Eve with multiple A/D's is briefly discussed in Section \ref{sec:7}.
The random gain is from a fixed probability distribution and is multiplied to the signal amplitude of each symbol that Alice transmits. 
Suppose that $A$ denotes  the random variable associated with this random gain, and  the probability  density function (pdf) of this gain is  $p_A(a)$ where $a \in \mathcal{A}$ (see Figure \ref{fig:p10}). 
The pdf of $A$ is known to all nodes, but only  legitimate nodes know the exact sequence of values of $A$ (i.e. $ a_1,a_2,a_3,\cdots$) that is applied to the symbol sequence. 

We want to find a probability distribution for $A$ that maximizes this secrecy rate  such that it does not change the average power of the transmitted signal, i.e. $\text{E}[|A|^2]=1$. 
To control the number of key bits required, we consider that $|A|$ is drawn from one of two levels $A_1$ and $A_2$ with random polarity (i.e. $\mathcal{A}=\{A_1,-A_1,A_2,-A_2\}$):
\[
Pr(A=a)=\left\{
\begin{array}{l l}
p, & \quad a=A_{1}\\
1-p,& \quad a=A_{2}\\
\end{array} \right.
\]
and $Pr\{A>0\}=Pr\{A<0\}=\nicefrac{1}{2}$. Suppose that $A_1$ is the large gain and $A_2$ is the small gain that the transmitter applies and denote the ratio between them $r=\frac{A_1}{A_2}$.

Since Bob shares the (long) key with Alice, he easily ``inverts'' the gain $A$ to operate his A/D properly, whereas Eve will struggle with such. 
In essence, we are inducing a fading channel at Bob that he is able to equalize \textit{before} his A/D, whereas Eve cannot. 
Bob applies the reciprocal of $A$ before his A/D and thus given $A$,  the signal that Bob's A/D sees is:
\begin{equation}
\tilde {Y}=X+\frac{n_B}{A}
\label{eq:1a}
\end{equation}
To cancel the effect of this gain, Eve also applies an arbitrary (possibly random) gain, $1/G$.
So, the signal at Eve's A/D given $A$ and $G$ is:
\begin{equation}
\tilde {Z}=\frac{A}{G}X+\frac{n_E}{G}
\label{eq:1b}
\end{equation}
Suppose that Eve knows the pdf of $A$; hence, she tries to find a probability density function $p_G(g)$ for $G$ such that it minimizes the secrecy rate $R_s$. 
On the other hand, Alice sets the pdf parameters  such that no matter what $p_G(g)$ Eve chooses, some secrecy rate $R_s$ is always guaranteed.
Hence, the maxi-min criteria in (\ref{eq:101}) turns into:
\begin{equation}
R_s=\max_{p, A_1, A_2}\min_{p_G(.)} \hat{R}_s(p_G(.),A_1,A_2,p)\label{eq:10}
\end{equation}
Obviously, larger $r=\frac{A_1}{A_2}$ leads to more eavesdropper confusion. However, because $E[|A|^2]=1$, $r \gg 1$ leads to a small $A_2$, and Bob then suffers noise enhancement. We talk about the choice of $r$ in the next paragraph.

Recall the potential operating scenario from Section \ref{sec:1}, and assume that system radios are operating in a scenario where they have adequate power amplifier headroom,
 as in the ``near'' situation in cellular systems \cite{kohno1995spread}, and the user's noise is relatively negligible. 
However, an Eve at the same range can also intercept the signal. 
By changing  the power of the transmitters  between the power-controlled level (e.g. $A_2$), where it meets the receiver requirements and its maximum power (e.g. $A_1$), Bob, knowing the sequence, obtains a signal that is at least equivalent to operating at its power controlled level and thus sees little  degradation in information transmission. The ratio between the large gain and the small gain, $r$, can be chosen such that in the case of $A=A_2$ (small gain), the minimum acceptable signal level at Bob's receiver is satisfied. 
On the other hand, Eve's A/D struggles even to record a reasonable form of the signal; hence, she sees significant degradation, and information-theoretic security is obtained. 
Also, because the power level is changed very fast (at every symbol), the automatic gain control (AGC) at the eavesdropper's receiver cannot follow the deep fades that cause erasures and/or strong signals that cause A/D saturation.

To choose optimum values for $A_1$, $A_2$, and $p$, note that the following constraints must be met:
\begin{equation}
\frac{A_1}{A_2}=r\quad \text{and}\quad p A_1^2+(1-p) A_2^2=1
\label{eq:cons}
\end{equation}
Hence, two of these values are constrained by the system parameter $r$ and conservation of transmission power, and the transmitter is free to choose only one (e.g. $p$).  Thus, equation (\ref{eq:10}) reduces to:
\begin{equation}
R_s=\max_{p}\min_{p_G(.)} \hat{R}_s(p_G(.),p)\label{eq:11}
\end{equation}
Eve can employ a number of countermeasures  to decrease $R_s$. 
She can find an optimum probability density function  that minimizes $R_s$, or  she can employ a better A/D to decrease erasures and/or overflows of her A/D.
In the  sequel, we will consider these scenarios and examine the  secrecy rate $R_s$ that can be achieved by the proposed method in each case.

\section{Achievable Secrecy Rates}\label{sec:3}
In this section the secrecy rates that can be achieved considering the non-idealities of the A/D's at the front-ends of  Bob and Eve's receivers are studied. In the first part, the channel between Alice and Bob and the channel between Alice and Eve are considered to be AWGN channels. In the second part, to get more insight into the problem, the noise is removed from the channels and only the effect of A/D's on the signals will be considered. 

\subsection{Noisy channels}
Consider the derivation of $I(X;Y|A=a)-I(X;Z|A=a,G=g)$. Clearly, each of $h(Y|A=a)$, $h(Y|X,A=a)$, $h(Z|A=a,G=g)$, and $h(Z|X,A=a,G=g)$ are required. 
Since for given gains at Alice and Eve, i.e. $A=a$ and $G=g$, by substituting $Z$ with $Y$ and $g$ with $a$ (Figure \ref{fig:p10}), the equations for $h(Y|A=a)$, $h(Y|X,A=a)$ can be derived from the equations for  $h(Z|A=a,G=g)$ and $h(Z|X,A=a,G=g)$, we just show the calculations for the latter here.
In this section all the mutual information, entropy, and probability density functions are calculated given that $A=a$ and $G=g$.

Recall that throughout this paper the non-idealities of the A/D's are modeled by an additive uniformly distributed quantization noise and a clipping function; hence,the signal at the output of Eve's A/D is:
\[
Z=\left\{
\begin{array}{l l}
\tilde{Z}+n_q, & \quad |\tilde{Z}|<l\\
+l, & \quad \tilde{Z}>l\\
-l, & \quad \tilde{Z}<-l\\
\end{array} \right.
\]
where $\tilde{Z}=\frac{aX}{g}+\frac{n_E}{g}$ and $l$ is determined by the range $[-l,l]$ of the A/D. 
Thus, $\tilde{Z}$ has a zero-mean Gaussian distribution with variance $\frac{a^2P+\sigma_E^2}{g^2}$, i.e. $\tilde{Z}\sim\mathcal{N}(0,\frac{a^2P+\sigma_E^2}{g^2})$.  Let us define the random variable $E'$ that takes the values $E'_1$, $E'_2$, and $E'_3$, where $E'_1=\{|\tilde{Z}|<l\}$ is the event that the signal before Eve's A/D falls in its dynamic range, and the events  $E'_2=\{\tilde{Z}>l\}$ and $E'_3=\{\tilde{Z}<-l\}$  correspond to clipping (A/D overflow).   We have,
\[
h(Z)=h(Z|E')+H(E')-H(E'|Z),
\]
Since $E$ is completely determined by $Z$, $H(E|Z)=0$; thus,
\begin{align*}
h(Z)&=\sum_{i=1}^3h(Z|E'_i)p(E'_i)-\sum_{i=1}^3 p(E'_i)\log (p(E'_i)).
\end{align*}
In the case of clipping we have $h(Z|E'_2)=h(Z|E'_3)=0$. The probability that the A/D is not in overflow is:
\[
p(E'_1)=1-2Q\left(\frac{gl}{\sqrt{a^2P+\sigma_E^2}}\right),
\]
and the probability that her A/D overflows is given by:
\[
p(E'_2)=p(E'_3)=Q\left(\frac{gl}{\sqrt{a^2P+\sigma_E^2}}\right),
\]
Then, $h(Z|E'_1)$ is calculated as:
\begin{align}
\nonumber &f_{Z|E'_1}(z)=f_{\tilde{Z}|E'_1}(z)*f_{n_q}(z)\\\nonumber
&=\frac{1}{\delta}\int_{-l}^{l}f_{\tilde{Z}}(s)U_{[-{\delta}/{2},{\delta}/{2}]}(z-s)ds\\\nonumber
&=\frac{1}{\delta}\int_{\max(-l,z-\delta/2)}^{\min(l,z+\delta/2)}f_{\tilde{Z}}(s)ds\\\nonumber
&\approx \frac{1}{\delta}\int_{z-\delta/2}^{z+\delta/2}f_{\tilde{Z}}(s)ds\\
&=\frac{1}{\delta}\left(Q\left(\frac{g(z-\delta/2))}{\sqrt{a^2P+\sigma_E^2}}\right)-Q\left(\frac{g(z+\delta/2)}{\sqrt{a^2P+\sigma_E^2}}\right)\right),\: |z|<l \label{eq:n4}
\end{align}
where $U_{[-{\delta}/{2},{\delta}/{2}]}(.)$ is the rectangle function on $[-{\delta}/{2},{\delta}/{2}]$, i.e. the value of the function is 1 on the interval  $[-{\delta}/{2},{\delta}/{2}]$ and is zero elsewhere. The reason that the approximation is valid is that we assume  high precision A/Ds are applied and thus $\delta \ll l$. Hence,
\begin{align}\nonumber
h(Z)=&\left(1-2Q\left(\frac{gl}{\sqrt{a^2P+\sigma_E^2}}\right)\right)\\
&\int_{-l}^l -f_{Z|E'_1}(z)\log(f_{Z|E'_1}(z)) dz +H(E').
\label{eq:n5}
\end{align}
Similarly, for  $h(Z|X)$ we have, 
\[h(Z|X)=h(Z|X,E')+H(E'|X)-H(E'|X,Z)
\]
Since $H(E'|X,Z)=0$,
\begin{align}
 h(Z|X)=\sum_{i=1}^3h(Z|E'_i,X)p(E'_i|X)+H(E'|X)
\label{eq:n7}
\end{align}
where $h(Z|E'_2,X=x)=h(Z|E'_3,X=x)=0$. The probability that Eve's A/D works in its dynamic range given $X$ is,
\begin{align*}
p(E'_1|X=x)&=p(|\tilde{Z}|<l|X=x)\\
&=p(|\frac{ax}{g}+\frac{n_E}{g}|<l)\\
&=p(-(gl+Ax)<n_E<gl-Ax))\\
&=Q\left(\frac{-(gl+Ax)}{\sigma_E}\right)-Q\left(\frac{gl-Ax}{\sigma_E}\right)
\end{align*}
and the probability that her A/D overflows,
\begin{align*}
p(E'_2|X=x)&=p(\tilde{Z}>l|X=x)\\
&=p(\frac{ax}{g}+\frac{n_E}{g}>l)\\
&=Q\left(\frac{gl-Ax}{\sigma_E}\right),
\end{align*}
and,
\begin{align*}
p(E'_2|X=x)&=p(\tilde{Z}<-l|X=x)\\
&=p(\frac{ax}{g}+\frac{n_E}{g}<-l)\\
&=Q\left(\frac{gl+Ax}{\sigma_E}\right).
\end{align*}

In order to calculate $h(Z|E'_1,X)$,  $f_{Z|E'_1,X=x}(z)$ is required. The signal before Eve's A/D $\tilde{Z}$ given $X=x$ has a Gaussian distribution with  mean $Ax/g$ and variance ${\sigma_E^2}/{g^2}$ within interval $|{ax}/{g}+{n_E}/{g}|<l$ and zero elsewhere. Hence,
\begin{align*}
&f_{Z|E'_1,X=x}(z)=f_{\tilde{Z}|E'_1,X=x}(z)*f_{n_q}(z)\\
&\approx\frac{1}{\delta}\int_{z-\delta/2}^{z+\delta/2} f_{\tilde{Z}|X=x}(z)ds\\
&=\frac{1}{\delta}\left(Q\left(\frac{g(z-\delta/2)-Ax}{\sigma_E}\right)-Q\left(\frac{g(z+\delta/2)-Ax}{\sigma_E}\right)\right),
\end{align*}
 for $|z|<l$, and,
\begin{align}
\nonumber
&h(Z|X)=\\ \nonumber
&\int_{-\infty}^{\infty}\Big[\int_{-l}^l  -f_{Z|E'_1,X=x}(z)\log(f_{Z|E'_1,X=x}(z))dz\: p(E'_1|X=x)\\
&\quad -\sum_{i=1}^3p(E'_i|X=x)\log(p(E'_i|X=x))\Big] f_X(x) dx
\label{eq:n6}
\end{align}
By substituting $h(Z)$ from (\ref{eq:n5}) and $h(Z|X)$ from (\ref{eq:n6}) in the following equation,
\begin{equation}
I(X;Z)=h(Z)-h(Z|X),\label{eq:n13}
\end{equation}
the mutual information between Alice and Eve given $A=a$ and $G=g$, can be found. 
Also, by substituting $Z$ with $Y$, $\sigma_E^2$ with $\sigma_B^2$, and $g$ with $a$ in (\ref{eq:n5}), (\ref{eq:n6}), and 
(\ref{eq:n13}), the mutual information between Alice and Bob given $A=a$ can be found,
\begin{equation}
I(X;Y)=h(Y)-h(Y|X)\label{eq:n12}
\end{equation}
The achievable secrecy rate can be found by substituting these mutual informations into the following equation:
\begin{align}
 R_s=E_{G,A}\left[I(X;Y)-I(X;Z)\right] \label{eq:9}
\end{align}

Alice is able to choose $p$ to maximize the $R_s$ that can be achieved by this method; 
on the other side, Eve tries to minimize $R_s$ by choosing an appropriate $p_G(.)$.
The following lemma shows that for an arbitrary discrete alphabet for $G$, choosing a single value  (which depends on the value of $p$) with probability one minimizes the secrecy rate, and thus is the optimal strategy for Eve.
\begin{lem}
The  gain $1/G$ that Eve applies before her A/D should take a single value with probability one to  minimize the secrecy rate.
\end{lem}
\begin{proof}
Suppose  $G$ has the following probability mass function:
\[
p_G(g=G_i)=\alpha_i,  \quad \quad  i=1,\cdots,n\\
\]
such that $\sum_{i=1}^{n}\alpha_i =1$. Without loss of generality, assume that for a specific $p$, the maximum information leakage occurs at $G=G_1$, i.e. for any gain $G_i, i=2,\cdots,n$ we have $I(X;Z|G=G_1)\geq I(X;Z|G=G_i)$; hence,
\begin{align*}
I(X;Z)&=\sum_{i=1}^{n}\alpha_i I(X;Z|G=G_i)\\
& \leq \sum_{i=1}^{n}\alpha_i I(X;Z|G=G_1)=I(X;Z|G=G_1)
\end{align*}
\end{proof}
The above lemma can easily be generalized to continuous random variables. Numerical results are given in Sections \ref{sec:4} and \ref{sec:5}.
\subsection{Noiseless Channels}
In the case that the channel between Alice and  Eve is noiseless, $h(Z)$ can be found by setting $\sigma_E^2=0$ in (\ref{eq:n5}). 
Using (\ref{eq:n6}) and the fact that $h(Z|E'_2,X=x)=h(Z|E'_3,X=x)=0$ and $H(E'|X)=0$ we have,
\begin{align}
\nonumber &h(Z|X)\\
\nonumber&=\int_{-\infty}^{\infty}h(Z|E'_1,X=x)p(E'_1|X=x)f_X(x)dx \\
\nonumber&=\int_{-\infty}^{\infty}h(\frac{aX}{g}+n_q|E'_1,X=x)p(E'_1|X=x)f_X(x)dx\\
\nonumber&=\int_{-Gl/|A|}^{gl/|A|}h(n_q)f_X(x)dx\\
&=\log(\delta)\left(1-2Q\left(\frac{gl}{a\sqrt{P}}\right)\right)
\label{eq:n10}
\end{align}
Similarly, in the case that Bob has a noiseless channel, 
\begin{equation}
h(Y|X)=\log(\delta)\left(1-2Q\left(\frac{l}{\sqrt{P}}\right)\right)\label{eq:n11}
\end{equation}

In each case, the secrecy rate can be found by substituting (\ref{eq:n10}) and (\ref{eq:n11}) in (\ref{eq:n13}) and (\ref{eq:n12}), respectively. Numerical results for the noiseless channels are shown in Sections \ref{sec:4} and \ref{sec:5}.
 \begin{figure}
\begin{center}
 \includegraphics[width=.4\textwidth]{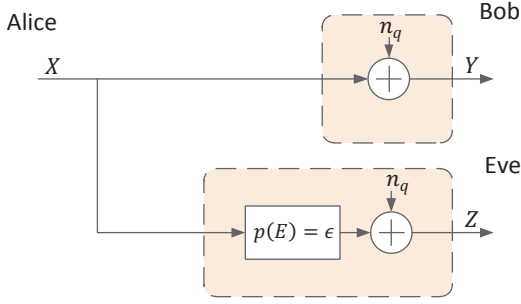}
 \end{center}
 \caption{Gaussian erasure wiretap channel: in the asymptotic case, the erasures/overflows at Eve's A/D due to the rapid power modulation at the transmitter can be modeled by an erasure channel. }
 \label{fig:pic5}
 \end{figure}

Clearly, considering  noiseless channels  makes the results less complicated and thus more insightful.
Hence, we continue our investigation by studying 
the asymptotic behavior of the proposed method (as  $r\rightarrow \infty$) in the noiseless regime, which will help us to achieve some intuition regarding this scheme. 
We assume that  Bob and Eve use  A/D's of the same quality for this analysis.
 
 Since in the noiseless regime $I(X;Y)$ does not depend on $A$, it does not change with  $r$ and thus we need only  evaluate $I(X;Z)$ for our asymptotic analysis.

From  (\ref{eq:cons})  we have,
\begin{equation}\label{eq:As}
A_1=\frac{r}{\sqrt{pr^2+(1-p)}}\quad\text{and}\quad A_2=\frac{1}{\sqrt{pr^2+(1-p)}}
\end{equation}
Let $G(r)$ be the inverse of the gain that Eve employs as a function of $r$.
Recall  from Lemma 1 that $G(r)$ will take a single value with probability one for a given $r$, but
that value can depend on $r$. Since $A_1\rightarrow  1/\sqrt{p}$ and $A_2\rightarrow 0$, we claim that in the limit (as $r\rightarrow \infty$), the best strategy that Eve can take is to choose either $G(r)=\Theta(1)$ or $G(r)=\Theta(r^{-1})$; otherwise, she will get no information (see Appendix A).
  
 First we study the secrecy rates that can be achieved when $G(r)=\Theta(1)$ as $r$ approaches $\infty$. The average secrecy rate is:
 \begin{align}
 \nonumber R_s&=E[I(X;Y)-I(X;Z)]\\ \nonumber
 &=p(I(X;Y|A=A_1)-I(X;Z|A=A_1))\\
 &\quad +(1-p)(I(X;Y|A=A_2)-I(X;Z|A=A_2))\label{eq:ER}
 \end{align}
 Assuming that Bob chooses the optimum range for his A/D, the maximum $I(X;Z|A=A_1)$ that Eve can achieve is $I(X;Y|A=A_1)$  and hence the first term in (\ref{eq:ER}) is zero. To evaluate the second term, putting $G(r)$ and $A=A_2$ in (\ref{eq:n5}) and (\ref{eq:n10}) yields:
 \begin{align}\nonumber
h(Z)=&\left(1-2Q\left(\frac{G(r)l}{A_2\sqrt{P}}\right)\right)\\
&\int_{-a}^a -f_{Z|E'_1}(z)\log(f_{Z|E'_1}(z)) dz+H(E')  
\label{eq:ne13}
\end{align}
where  since $G(r)=\Theta(1)$, $\left(1-2Q\left(\frac{G(r)l}{A_2\sqrt{P}}\right)\right)\to 1$ as $r\to \infty$ and thus $H(E')\to 0$; and, for $|z|<l$,
\begin{align*}
&f_{Z|E'_1}(z)\\
&=\frac{1}{\delta}\left(Q\left(\frac{G(r)(z-\delta/2)}{A_2\sqrt{P}}\right)-Q\left(\frac{G(r)(z+\delta/2)}{A_2\sqrt{P}}\right)\right) \\
&\rightarrow \left\{
\begin{array}{l l}
\frac{1}{\delta}, & \quad  0<|z|<\delta/2\\
\frac{1}{2\delta}, & \quad  |z|=\delta/2\\
0, & \quad \quad \text{otherwise}\\
\end{array} \right.
\end{align*}
 Since   the integrand in (\ref{eq:ne13}) is bounded  for all $r$, from the dominated convergence theorem, $h(Z)\rightarrow \log\delta$ as $r\rightarrow \infty$.
Also since $G(r)=\Theta(1)$,
\begin{align}
 h(Z|X)&=\log(\delta)\left(1-2Q\left(\frac{G(r)l}{A_2\sqrt{P}}\right)\right)
\rightarrow \log(\delta)
\label{eq:n14}
\end{align}
as $r$ approaches $\infty$. Thus,  $I(X;Z|A=A_2)=0$ and hence the average secrecy rate given that $G(r)=\Theta(1)$ is $R_s=(1-p)I(X;Y)$.

Now suppose $G(r)=\Theta(r^{-1})$ and consider the second term in (\ref{eq:ER}).
In the limit,  ${A_2}/{G(r)}=c$ where $c>0$ is a bounded constant. 
Since Bob chooses the optimum range for his A/D, the maximum  $I(X;Z|A=A_2)$ that Eve can achieve is $I(X;Y|A=A_2)$ and thus given that $G(r)=\Theta(r^{-1})$, the second term in (\ref{eq:ER}) is zero. To evaluate  the first term in (\ref{eq:ER}) as $r$ gets large, by substituting   $G(r)=\Theta(r^{-1})$ and $A=A_1$ in (\ref{eq:n5}) and (\ref{eq:n10}), we have $f_{Z|E'_1}(z)\rightarrow 0$ and,
\[
\left(1-2Q\left(\frac{G(r)l}{A_1\sqrt{P}}\right)\right)\rightarrow 0
\] as $r$ approaches infinity and hence $h(Z)\rightarrow 0$. Also by letting $G(r)=\Theta(r^{-1})$ we have,
\begin{align*}
 h(Z|X)&=\log(\delta)\left(1-2Q\left(\frac{G(r)l}{A_1\sqrt{P}}\right)\right)\rightarrow 0 \quad\text{as} \quad r\rightarrow\infty
\end{align*}
Hence, with   probability $p$ the mutual information between Alice and Eve is zero and  the average secrecy rate that can be achieved given $G(r)=\Theta(r^{-1})$ as $r$ approaches $\infty$ is $R_s=pI(X;Y)$. 

We can interpret these results as follows;
when ${A}/{G(r)}={A_1}/{\Theta(r^{-1})}$, the total gain that Eve's A/D sees approaches infinity as $r \rightarrow \infty$; hence, even if Eve uses an A/D with larger range than Bob's A/D, her quantizer  overflows. 
 When ${A}/{G(r)}={A_2}/{\Theta(1)}$, the total gain goes to zero as $r$ approaches infinity and thus even if Eve uses an A/D with better precision, the received signal amplitude is less than one quantization level. In both cases, Eve receives no information about the transmitted signal and thus Eve's channel can be modeled by an erasure channel (Figure \ref{fig:pic5}), where for $G(r)=\Theta(r^{-1})$, the probability of erasure $\epsilon=1-p$ and for  $G(r)=\Theta(1)$,  $\epsilon=p$. 
  
 Hence, the secrecy rate that can be achieved in the asymptotic case (as $r\rightarrow \infty$) is:
 \begin{equation}
 R_s=(1-\epsilon) I(X;Y)
 \label{eq:Rs}
 \end{equation}
 \begin{figure}
\begin{center}
 \includegraphics[width=.5\textwidth]{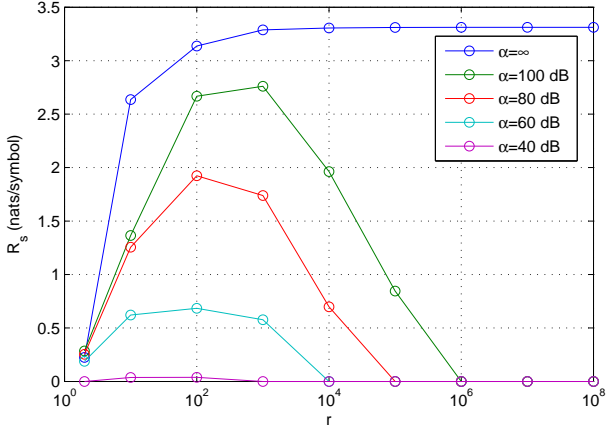}
 \end{center}
 \caption{Achievable secrecy rate versus $r$ (the ratio between the large and the small gain);  both Bob and Eve apply 10-bit A/D's with the dynamic range $l=2.5$.  The SNRs of both Bob's channel and Eve's channel are the same and are denoted by $\alpha$.}
 \label{fig:f7}
 \end{figure} 

   To maximize the achievable secrecy rate, it is reasonable for Alice to choose $p=0.5$. In Section \ref{sec:s1} it is shown that for a 10-bit A/D and the transmitter power $P=1$, the optimum range of the A/D is obtained by setting $l=2.5$, and the corresponding mutual information between Alice and Bob (when the channel between them is noiseless) is $I(X;Y)=6.597$. Hence, using (\ref{eq:Rs}), $R_s\rightarrow 0.5\times 6.597=3.2985$. Figure \ref{fig:f7} (the upper curve) shows the achievable secrecy rate versus $r$ when both Bob's channel and Eve's channel are noiseless. It can be seen that as $r$ gets larger, the achievable secrecy rate goes to a constant which is similar to what is anticipated. Furthermore,  for larger $r$'s ($r\geq10^3$) the  optimum probability that maximizes the worst case secrecy rate is $p=0.5$. 
   These results  show that  our results are consistent to expectations in the limit.
  
 From another point of view, consider that for small values of $\delta$, the quantization noise can be modeled by a zero mean Gaussian random variable with the variance ${\delta^2}/{12}$, where $\delta$ is the size of each quantization level. Thus, this wiretap channel can be modeled by a Gaussian erasure wiretap channel.
   
The secrecy capacity of the Gaussian wiretap channel is \cite{barros2006secrecy}:
\[C_s=\frac{1}{2}\left(\log(1+|h_B|^2\gamma_B)-\log(1+|h_E|^2\gamma_E)\right)^+\]
where $h_B$ and $h_E$ are channel gains, $\gamma_B$ is the SNR at Bob's receiver, and  $\gamma_E$ is the SNR at Eve's receiver. We can use this secrecy capacity in our asymptotic model by setting $h_B=1$ and  modeling the  erasure channel by an unusual fading channel with the following  fading distribution:
\[h_E=\left\{
\begin{array}{l l}
0,&\quad \text{w.p.}\quad \epsilon\\
1,&\quad \text{w.p.}\quad  1-\epsilon\\
\end{array} \right.
\]
Since we assumed that  Eve's A/D is identical to Bob's A/D, $\gamma_E=\gamma_B=\frac{P}{\nicefrac{\delta^2}{12}}$ and thus the secrecy capacity is non-zero only when an erasure at Eve's channel occurs. 
Hence,
\begin{equation}
C_s=\frac{(1-\epsilon)}{2} \log(1+\gamma_B) 
\label{eq:Cs}
\end{equation}
This equation shows that for a 10-bit A/D with $l=2.5$, transmitting power $P=1$, and $\epsilon=p=0.5$, the secrecy capacity is $C_s=3.2822$ which is again very close to what we expect from our asymptotic analysis. Furthermore, on comparing equations (\ref{eq:Rs}) and (\ref{eq:Cs}),  it is seen that in the asymptotic case, the achievable secrecy rate meets this approximate secrecy capacity.

\section{Numerical Results} \label{sec:4}

\subsection{Motivation}\label{sec:41}
\begin{figure}
\begin{center}
 \includegraphics[width=.5\textwidth]{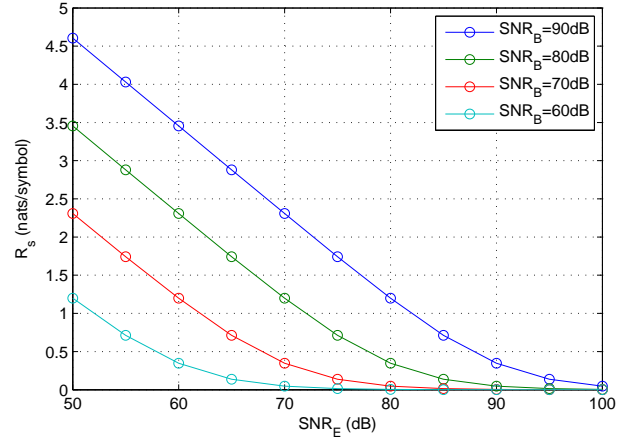}
 \end{center}
 \caption{Secrecy capacity of public discussion for various values of SNR at Bob's receiver when the SNR at Eve's receiver changes from 50 dB to 100 dB and  $P=1$. When the SNR at Bob's receiver is less than the SNR Eve's receiver, the secrecy rate drops rapidly.}
 \label{fig:p5}
 \end{figure} 
When the channel between Alice and Eve is less noisy than the channel between Alice and Bob, if the legitimate users are restricted to one-way and rate-limited communication, the secrecy capacity of the wiretap channel is zero.
However, if we relax the restrictions placed on the  schemes that the legitimate users can apply by allowing two-way communication and the presence of  a noiseless, public, and authenticated channel, public discussion strategies \cite{maurer1993secret,ahlswede1993common} allow the legitimate nodes to  agree on a secret key by extracting information from realizations of correlated random variables. 
This secret-key can then be used in a one-time-pad for secret communication between Alice and Bob. 
A closed form for the general secret-key capacity is not available; however, in the case of a Gaussian source model in which $X\sim\mathcal{N}(0,P)$ and a Gaussian wiretap channel, i.e. when the channel between Alice and Bob and the channel between Alice and Eve are AWGN channels, the secrecy capacity has a simple form \cite[Chapter 5]{bloch2011physical}:
\begin{equation}
C^{\text{SM}}_s=\frac{1}{2}\log\left(1+\frac{P\sigma_E^2}{(P+\sigma_E^2)\sigma_B^2}\right)
\end{equation}
and thus all secret-key rates less than $C^{\text{SM}}_s$ are achievable.
Achievable secrecy rates of  public discussion for various values of the signal-to-noise ratio at Bob's receiver versus signal-to-noise ratio at Eve's receiver are shown in Figure \ref{fig:p5}. 
As can be seen, when  the SNR of Eve's receiver is significantly larger than the  SNR at Bob's receiver, the secrecy rate of public discussion drops rapidly. Our main goal here is to see  whether our scheme can improve the performance in this regime.

\subsection{Noiseless Channels: Eve with the same A/D as Bob}\label{sec:s1}
 \begin{figure}
\begin{center}
 \includegraphics[width=.39\textwidth]{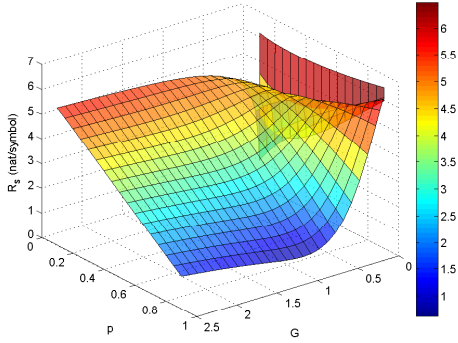}
 \end{center}
 \caption{Achievable secrecy rate vs. the probability $p$ and the gain $G$ at Eve's receiver. Both Bob's and Eve's channels are noiseless and they use identical 10-bit A/D's. The ratio between the two power levels at the transmitter is $r=10^3$ (i.e. 30 dB) and the average transmitting power is $P=1$. A maxi-min rate of $R_s=3.1372$ is achieved.}
 \label{fig:p2}
 \end{figure} 
We begin our investigation by considering only the effect of A/D's on the signals. 
Hence, we  assume that Eve's channel is noiseless, i.e. $n_E=0$ (which benefits the eavesdropper).
However,  we also assume the system nodes are working in a very high SNR regime and thus the channel noise at Bob can  be neglected ($n_B=0$).

Now suppose that both Bob and Eve use 10-bit quantizers ($b=2^{10}$) and the transmitter power is $P=1$. 
Since $\delta=2a/b$, for a fixed number of quantization bits, $I(X;Y)$ is a function of the 	 of the A/D ($a$), and the optimal quantization range that maximizes  $I(X;Y)$ can be found. Since $I(X;Y)$ is an intricate function in terms of $a$, we find the optimum $a$ numerically. 
In this case, the optimum quantization range that maximizes $I(X;Y)$  is $l=2.5$, and the corresponding mutual information between Alice and Bob is $I(X;Y)=6.597$. 
For the remainder of the paper, we use $l=2.5$ in our calculations.
Suppose that Eve has the same A/D as Bob. 
From Lemma 1,  putting a random gain is undesirable for Eve; hence, she  chooses a fixed gain $G$ that minimizes $R_s$. 
Because Alice is not aware of Eve's choice, she has to choose a probability $p$ that maximizes the worst case $R_s$. 

As we discussed in Section \ref{sec:2}, a larger $r$ leads to more eavesdropper confusion and thus as $r$ increases, the secrecy rate  would be expected to increase. However, in the case of noisy channels, a large $r$ also causes noise enhancement at  Bob's receiver that decreases the secrecy rate. 
In order to  get some insight about the dependency of the  secrecy rate on $r$,  curves of $R_s$ versus $r$ are  shown in Figure \ref{fig:f7}. For each curve, the SNR at both Eve's receiver and Bob's receiver are the same and are denoted by $\alpha$. 
Hence, in order to achieve high secrecy rates by avoiding excessive noise enhancement at Bob's receiver, for the rest of the paper we set $r=10^3$.
The plot of $R_s$ versus $p$ and $G$ for $P=1$ and $r=10^3$ (i.e. 30 dB) where both Bob and Eve are each using a 10-bit A/D is shown in Figure \ref{fig:p2}. This function is complicated and hence the optimum value of $p$ cannot be derived analytically. 
Numerical analysis shows that $p\approx 0.45$ maximizes the worst case $R_s$, and the maxi-min value is ${R_s}=3.1366$. Hence, choosing $p=0.45$ guarantees that at least the secrecy rate ${R_s}=3.1366$ can be achieved. 
\subsection{Noiseless Channels: Eve with a Better A/D than Bob}
\begin{figure}
\begin{center}
 \includegraphics[width=.39\textwidth]{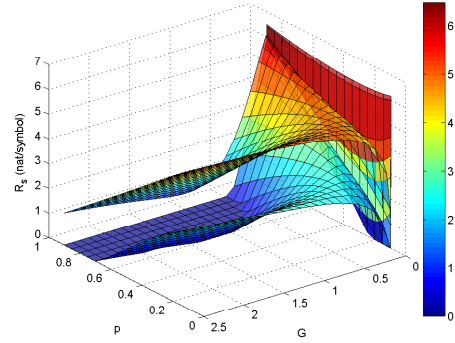}
 \end{center}
 \caption{Achievable secrecy rate vs. the probability $p$ and the gain at Eve's receiver, $G$ for the case of noiseless channels. The ratio between the two power levels at the transmitter is $r=10^3$ (i.e. 30 dB) and the average transmitting power is $P=1$. In the upper curve, both Bob and Eve have the same 10-bit A/D's. In the lower curve, Bob uses a 10-bit A/D while Eve uses a 14-bit A/D (Eve's A/D is 24 dB better than Bob's A/D) and a maxi-min rate of $R_s=1.2478$ is achieved (for $p=0.4$). }
 \label{fig:p3}
 \end{figure}
 Now suppose that Eve has access to a better A/D than Bob. 
Depending on the gain that Eve applies before her A/D, a better A/D results in less erasures and/or less A/D overflows. Hence, the mutual information between Alice and Eve increases and consequently, the achievable secrecy rate decreases.
Figure \ref{fig:p3} shows this effect  versus $p$ and $G$. It can be seen that even if Eve uses an A/D which is 24 dB (4 bits) better than Bob's A/D (Eve has a 14-bit A/D while Bob has a 10-bit A/D), by choosing an appropriate value for $p$, a positive secrecy rate can be achieved. 
In this example, by choosing $p=0.4$, a secrecy rate $R_s=1.2426$ is achievable. 
Even if we do not change the probability $p$ from the previous section ($p=0.45$), assuming that Alice is not aware of Eve's better A/D,  a secrecy rate $R_s=0.9225$ is achievable. 
In spite of having a better A/D, Eve will still lose some symbols and hence a positive secrecy rate is available. 
This is because the ratio between the large and the small gain,  $A_1$ and $A_2$, is $10^3$, while Eve's A/D has only 16 times better resolution; 
thus, she still needs to compromise between resolution and overflow. 
To cancel the effect of these gains completely, Eve has to use an A/D that has an effective  resolution after taking into account jamming, interference, etc. on the order of $10^3$ times (10 bits) better than Bob's A/D, which would be very difficult in an adversarial environment. 
 \subsection{Noisy Main Channel, Noiseless Eavesdropper's channel}\label{sec:s3}
\begin{figure}
\begin{center}
 \includegraphics[width=.5\textwidth]{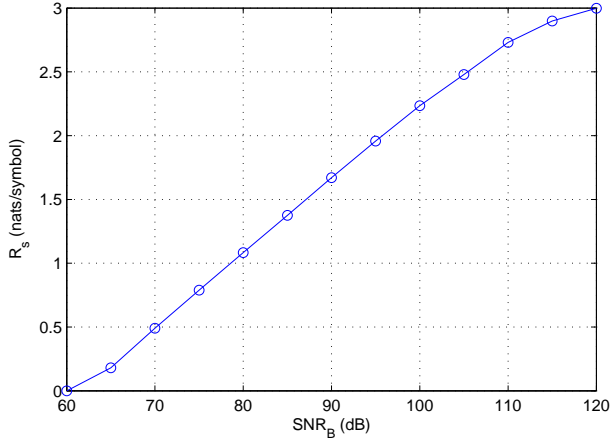}
 \end{center}
 \caption{Achievable secrecy rate vs. SNR at Bob's receiver while the SNR at Eve's receiver is infinity (Eve has perfect access to the transmitted signal) for $r=10^3$, $P=1$ and  Bob and Eve applying 10-bit A/D's.  Note that the assumption of Eve having a  noiseless channel is the extreme case when Eve has perfect access to the transmitter's output (for instance, the eavesdropper is able to pick up the transmitter's radio) and hence  other secrecy methods  are not effective. Using the proposed method, a positive secrecy rate can be achieved over short range at reasonable power.}
 \label{fig:p4}
 \end{figure} 
  Now we look at the extreme case that Eve is able to receive exactly what Alice transmits and receives (e.g. the adversary is able to pick up the transmitter's radio and hook directly to the antenna), but the channel between Alice and Bob is noisy and hence no other technique is effective. 
 In other words,  the channel between Alice and Bob experiences an additive white Gaussian noise ($n_B\sim\mathcal{N}(0,\sigma_B^2)$), while Eve's channel is noiseless ($n_E=0$). Figure \ref{fig:p4} shows the secrecy rate $R_s$ that can be achieved using the proposed scheme versus the signal-to-noise ratio (SNR) at Bob's receiver. In this case, the transmitted power $P=1$, the ratio between the large and the small gain is 30 dB, and both Bob and Eve use 10-bit A/D's. It can be seen that, although Eve's channel is much better than Bob's channel, when the SNR at Bob's receiver is greater than 60 dB, which could be made common in a short-range application as described in Section \ref{sec:1}, a positive secrecy rate is available. 
 By comparing the noise-free result in Figure \ref{fig:f7} for $r=10^3$ and Figure \ref{fig:p4}, it can be seen that the secrecy rate when SNR at Bob is 120 dB is still less than the secrecy rate when Bob's channel is noiseless.
 
\subsection{Noisy Channels}\label{sec:5}
%
When both channels are noisy, the achievable secrecy rate of  the proposed method  versus the SNR at Eve's receiver  for various values of the SNR at Bob's receiver is shown in Figures \ref{fig:p6}. 
The transmitted power $P=1$, the ratio between the large and the small gain is 30 dB, and both Bob and Eve use 10-bit A/D's.
It can be seen that by applying the proposed method for the case of Eve  with a (significantly) better channel than Bob, which is the regime of interest per Figure \ref{fig:p5}, reasonable secrecy rates can be achieved. 
Note that in our method we are generating an advantage for the legitimate nodes to be used with wiretap coding, and thus, because public discussion approaches assume the presence of a public authenticated channel, public discussion  should not be viewed as a competitor to the proposed scheme. 
Rather, if such a public authenticated channel exists and  two-way communication is possible, our method can be used \textit{in conjunction} with public discussion techniques to result in higher secrecy rates.
Nevertheless, per Figure \ref{fig:p5}, public discussion provides motivation for the regime where advances are needed given the current state of the art.
 \begin{figure}
\begin{center}
 \includegraphics[width=.5\textwidth]{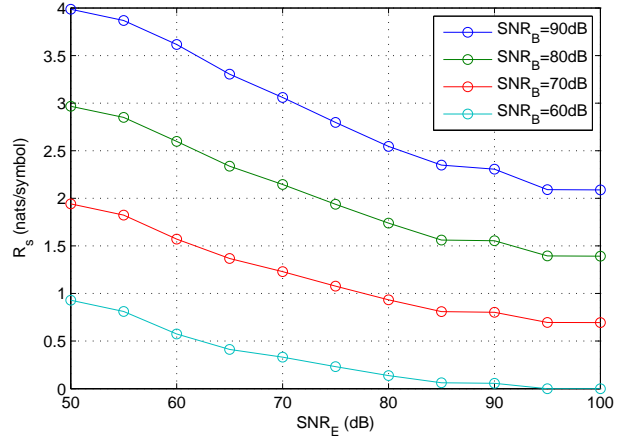}
 \end{center}
 \caption{Achievable secrecy rates for various values of the SNR at Bob's receiver when the SNR at Eve's receiver changes from 50 dB to 100 dB. The settings are $r=10^3$, $P=1$, and  Bob and Eve are applying 10-bit A/D's. When the SNR of the channel between Alice and Eve is significantly better than the SNR of the channel between Alice and Bob,  reasonable secrecy rates are still achievable.}
 \label{fig:p6}
 \end{figure} 
 \section{Conclusion}\label{sec:7}
In this paper, we introduce a new approach that exploits a short-term cryptographic key to 
force different orderings at Bob and Eve of two operators, one of which is necessarily non-linear, 
to obtain the desired advantage for information-theoretic security in a wireless communication
system regardless of the location of Eve. We then investigate a simple power modulation instantiation of the approach.  It is shown that when Eve's channel condition is significantly  better than the Bob's channel, reasonable secrecy rates   can still be achieved using our proposed method in this challenging regime. 
In particular,  even in the case that the adversary is able to pick up the transmitter's radio (i.e. Eve has perfect access to the output of the transmitter), a reasonable secrecy rate is achievable at
high SNRs which might apply to a short-range wireless system.  For example, one might use
the transmission power of typical cellular systems with the corresponding excess power at short
ranges to establish a secure radio system in a limited area.

Although we have considered the case of Eve with a better A/D than Bob, the clear risk to the
approach is still that of asymmetric capabilities at the receivers.  For example, if we employ the
simple power modulation approach studied extensively here, Eve may employ
multiple A/D's with different gain settings in front of each.  Hence, Eve would be able to record two signals independently and decode them later when she gets the key or extracts the key based on the pattern of erasures and overflows at each A/D.   A simple approach to combat this attack is rather than applying just two power levels, the transmitter can apply many power levels.  
More promising, however, is to consider adding memory to the signal warping
  process \cite{isit2013}.

Broadly considering potential techniques for everlasting security in wireless systems, including 
that proposed here, yields that each approach still holds some risk.  In the case of 
cryptographic security, assumptions must be made on both the hardness of the 
problem and the current/future computational capabilities of the adversary.  In the case of
standard information-theoretic security, assumptions must be made on the quality of the channel
to Eve, generally corresponding to limitations on her location.  In the method proposed
here, assumptions must be made on Eve's current conversion hardware capabilities, but,
as in standard information-theoretic secrecy, there is no assumption on future capabilities.  All three
approaches thus have different applicability.
\section*{Appendix A}
In this section we show that as $r\to\infty$ the only strategy that Eve can take to obtain information from the signal she receives  is to choose either $G(r)=\Theta(1)$ or $G(r)=\Theta(r^{-1})$.
Instead of applying  $G(r)=\Theta(1)$ or $G(r)=\Theta(r^{-1})$, the two other possibilities for Eve are to choose $G(r)$ such that either  $\lim_{r \to \infty} r^{-1}/G(r)\to 0$ or  $\lim_{r \to \infty} r^{-1}/G(r)\to \infty$  (and obviously provided that $G(r)\neq\Theta(1)$). 

First suppose $\lim_{r \to \infty} r^{-1}/G(r)\to 0$ and consider  $I(X;Z|A=A_1)$ in (\ref{eq:ER}). Since $G(r)\neq\Theta(1)$ and from (\ref{eq:As}),  $\lim_{r \to \infty} A_1/G(r)\to 0$ and hence,
\begin{align}
h(Z)=\left(1-2Q\left(\frac{G(r)l}{A_1\sqrt{P}}\right)\right)\int_{-a}^a -f_{Z|E'_1}(z)\log(f_{Z|E'_1}(z)) dz  
\label{eq:nee13}
\end{align}
where, for $|z|<l$,
\begin{align*}
f_{Z|E'_1}(z)&=\frac{1}{\delta}\left(Q\left(\frac{G(r)(z-\delta/2)}{A_1\sqrt{P}}\right)-Q\left(\frac{G(r)(z+\delta/2)}{A_1\sqrt{P}}\right)\right) \\
&\rightarrow \left\{
\begin{array}{l l}
\frac{1}{\delta}, & \quad  0<|z|<\delta/2\\
\frac{1}{2\delta}, & \quad  |z|=\delta/2\\
0, & \quad \quad \text{otherwise}\\
\end{array} \right.
\end{align*}
and $\left(1-2Q\left(\frac{G(r)l}{A_1\sqrt{P}}\right)\right)\rightarrow 1$ as $r\rightarrow \infty$. Since   the integrand in (\ref{eq:nee13}) is bounded  for all $r$ and from the dominated convergence theorem, $h(Z)\rightarrow \log\delta$ as $r\rightarrow \infty$.
Also, since $\lim_{r \to \infty} r^{-1}/G(r)\to 0$,
\begin{align}
 h(Z|X)&=\log(\delta)\left(1-2Q\left(\frac{G(r)l}{A_1\sqrt{P}}\right)\right)
\rightarrow \log(\delta)
\label{eq:500}
\end{align}
as $r$ approaches $\infty$ and thus  $I(X;Z|A=A_1)=0$. Now consider  $I(X;Z|A=A_2)$ in (\ref{eq:ER}); by substituting $A_1$ with $A_2$ in (\ref{eq:nee13}) and (\ref{eq:500}), and since $\lim_{r \to \infty} A_2/G(r)\to 0$,  we have $I(X;Z|A=A_2)=0$. Consequently, given that $\lim_{r \to \infty} r^{-1}/G(r)\to 0$,  the average information that Eve obtains is zero.

Now suppose  $\lim_{r \to \infty} r^{-1}/G(r)\to \infty$ and consider the first term $I(X;Z|A=A_1)$ in (\ref{eq:ER}).  The fact that $\lim_{r \to \infty} r^{-1}/G(r)\to \infty$ implies that  in the limit as $r\to\infty$, $ A_1/G(r)$ also goes to $\infty$ and thus from (\ref{eq:n5}) and (\ref{eq:n10}) we have $f_{Z|E'_1}(z)\rightarrow 0$. Also,
$\left(1-2Q\left(\frac{G(r)l}{A_1\sqrt{P}}\right)\right)\rightarrow 0$ as $r$ approaches infinity and hence $h(Z)\rightarrow 0$. Furthermore,
\begin{align}\label{eq:hzx}
 h(Z|X)&=\log(\delta)\left(1-2Q\left(\frac{G(r)l}{A_1\sqrt{P}}\right)\right)\rightarrow 0 
\end{align}
as $r\to\infty$ and thus $I(X;Z|A=A_1)=0$. Considering  $I(X;Z|A=A_2)$ in (\ref{eq:ER}) and by putting $A_2$ instead of $A_1$ in (\ref{eq:hzx}), since   $A_2/G(r)\rightarrow \infty$ in the limit as $r\rightarrow \infty$, we have $I(X;Z|A=A_2)=0$. Hence, by choosing $\lim_{r \to \infty} r^{-1}/G(r)\to \infty$ Eve gets no information about the transmitted signal.

\bibliographystyle{ieeetr}
\bibliography{mycite}

\end{document}